\documentclass[fleqn]{eptcs}
\usepackage{amsmath}
\usepackage{amssymb}
\usepackage{wasysym}
\usepackage{graphics}
\usepackage[usenames]{color}
\usepackage{macros_art}
\usepackage{pstricks,pst-node}
\usepackage{gastex}

\usepackage{MnSymbol}

\def\true{\mathit{true}}

\def\false{\mathit{false}}

\def\FFF{\mathcal{F}}

\def\supp{\mathsf{supp}}

\DeclareMathOperator{\untl}{\UUU}
\DeclareMathOperator{\wkuntl}{\WWW}

\DeclareMathOperator{\dmnd}{\Diamond}

\newcommand{\caoal}[1][A]{{\lsem}{\mathit #1}{\rsem}}
\newcommand{\ceoal}[1][A]{{\llangle}{\mathit #1}{\rrangle}}
\newcommand{\caggg}[1][A]{\caoal[#1]\!\Box}
\newcommand{\ceggg}[1][A]{\ceoal[#1]\!\Box}
\newcommand{\cafff}[1][A]{\caoal[#1]\!\!\Diamond}
\newcommand{\cefff}[1][A]{\ceoal[#1]\!\!\Diamond}
\newcommand{\canxt}[1][A]{\caoal[#1]\!\!\ocircle\!}
\newcommand{\cenxt}[1][A]{\ceoal[#1]\!\!\ocircle\!}

\def\ATLK{ATL$^{\text{D}}_{iR}$}

\def\runsfin{\mathsf{Runs^f}}
\def\runsinf{\mathsf{Runs^\omega}}

\newcommand{\oomit}[1]{}

\def\sdpart{\rightharpoonup}

  \newtheorem{definition}{Definition}
  \newtheorem{theorem}[definition]{Theorem}
  \newtheorem{proposition}[definition]{Proposition}
  \newtheorem{remark}[definition]{Remark}
  
  \newtheorem{example}[definition]{Example}
  
 \newenvironment{proof}{\smallskip\emph{Proof:}}{\smallskip$\dashv$}

\def\titlerunning{Model-checking \ATLK{}}
\def\authorrunning{C. Dima, C. Enea \& D. Guelev}

\title{Model-Checking an Alternating-time Temporal Logic with Knowledge, Imperfect Information, Perfect Recall
and Communicating Coalitions\thanks{Partially supported by the French National Research Agency (ANR) projects SELKIS 
(ANR-08-SEGI-018) and POLUX (ANR-06-SETIN-012).}}
\author{
C\u at\u alin Dima
 \institute{
LACL, Universit\'e Paris Est-Cr\'eteil, 
61 av. du G-ral de Gaulle, 94010 Cr\'eteil, France}
\and 
Constantin Enea\thanks{This author was partly supported by the French National Research Agency (ANR) projects Averiss (ANR-06-SETIN-001) and Veridyc (ANR-09-SEGI-016).}
\institute{LIAFA, CNRS UMR 7089, Universit\'e Paris Diderot - Paris 7, 
Case 7014, 75205 Paris Cedex 13, France}
\and 
Dimitar Guelev\thanks{This author was partly supported by Bulgarian National Science Fund Grant ID-09-0112. Work partly done while he 
was visiting the Universit\'e Paris-Est Cr\'eteil. }
\institute{Section of Logic, Institute of Mathematics and Informatics, 
Bulgarian Academy of Sciences,\\
 Acad. G. Bonchev str., bl. 8, 1113 Sofia, Bulgaria } 
}

\begin{document}

\maketitle

\begin{abstract}
\textbf{Abstract.} We present a variant of ATL with distributed knowledge 
operators based on a synchronous and perfect recall semantics. 
The coalition modalities in this logic are based on partial observation
of the full history, and incorporate a form of
cooperation between members of the coalition in which 
agents issue their actions based on the distributed knowledge, for that coalition, 
of the system history.
We show that model-checking is decidable for this logic.
The technique utilizes two variants of games with imperfect information 
and partially observable objectives, as well as a subset construction 
for identifying states whose histories are 
indistinguishable to the considered 
coalition.
\end{abstract}

\section{Introduction}

Alternating-time Temporal Logic (ATL) \cite{alur-atl1997,alur-atl-jacm2002} is a
generalization of the Computational Tree Logic (CTL) in which path
quantifiers ``$\exists$'' and ``$\forall$'' are replaced by
\emph{cooperation modalities} $\ceoal$ in which $A$ denotes a set of
{\em agents} who act as a {\em coalition}.  A formula $\ceoal\phi$
expresses the fact that the agents in {\em coalition} $A$ can
cooperate to ensure that $\phi$ holds in an appropriate type of
multiplayer game.  

The precise semantics of the cooperation modalities
varies depending on whether the knowledge that each agent has of the
current state of the game is complete or not, and whether agents can
use knowledge of the past game states when deciding on their next move
or not. These alternatives are known as {\em complete}, resp. {\em
  incomplete information}, and {\em perfect}, resp. {\em imperfect
  recall}. In the case of imperfect recall further subdivisions depend
on how much memory an agent is allowed for storing information on the
past in addition to its possibly incomplete view of the current state.
In the extreme case agents and, consequently, the strategies they can
carry out, are {\em memoryless}.  

It is known that the model-checking
problem for the case of complete information is decidable in polynomial time
\cite{alur-atl1997}.  In the case of incomplete information and perfect
recall model-checking is believed to be undecidable, a statement 
attributed to M. Yannakakis in \cite{alur-atl1997} for which there is no
self-contained proof 
that we know about.  Variants of ATL with
memoryless agents have been shown to have decidable model checking in
\cite{schobbens-atl-ir2004,goranko-irrevocable2007,lomuscio-practicalATL2006}. 
Our earlier work \cite{guelev-dima2008dalt} is about
a special case of agents with perfect recall in which model checking
is still decidable.  

Incomplete information is modelled in ATL in a
way which conforms with the possible worlds semantics of modal
epistemic logics (cf. \cite{FaginHalpernVardi}.)  Therefore, it is of no
surprise that the epistemic logic community contributed extensions of
ATL by knowledge modalities such as \emph{Alternating Temporal Epistemic Logic} 
\cite{wiebe-ATEL2003}.
Results on model-checking ATEL with memoryless strategies can be found in \cite{schobbens-atl-ir2004,goranko-irrevocable2007,kacprzak-penczek-aamas-2005,lomuscio-practicalATL2006}. Results on ATL with complete information can be found in \cite{goranko-jamroga-comparing-2004,bulling-harder-2009}.
\oomit{Yet, only partial results are
known about the model-checking problem for ATEL, some of them applying
to memoryless strategies 
\cite{schobbens-atl-ir2004,goranko-irrevocable2007,kacprzak-penczek-aamas-2005,lomuscio-practicalATL2006},
some others to strategies with complete observations \cite{goranko-jamroga-comparing-2004,bulling-harder-2009}.
(At the end of this section, we comment on the relationship between these results and 
ours.)}

In this paper we continue our investigation of ATL with knowledge
operators from \cite{guelev-dima2008dalt}, where we introduced conditions on
the meaning of the cooperation modalities which make model-checking
decidable. As in the previous paper, we do not restrict agents'
strategies to memoryless ones, but we assume that coalition members have 
a communication mechanism which enables the coalitions to  carry out strategies 
that are based on their {\em distributed knowledge}. 
(Recall that a coalition has \emph{distributed knowledge} of fact $\phi$ 
iff $\phi$ is a logical consequence of the combined knowledge of the coalition members.) 
We assume that a coalition has a strategy to achieve a goal $\phi$ 
only if the same strategy can be used in all the cases which are indistinguishable from the actual one
with respect to the distributed knowledge of the coalition. 
This choice is known as {\em de re} strategies \cite{jamroga-de-dicto2007}, and rules
out the possibility for a coalition to be able to achieve $\phi$ by
taking chances, or to be able to achieve $\phi$ in some of the cases
which are consistent with its knowledge and not in others. Therefore
in our system $\ceoal\phi$ is equivalent to $K_A\ceoal\phi$ where
$K_A$ stands for the \emph{distributed knowledge} operator (also written
$D_A$). 
We call the variant of ATL which is obtained by adopting these conventions 
{\em Alternating Time Logic with Knowledge and Communicating Coalitions} and use the acronym \ATLK{} for it to indicate distributed knowledge, incomplete information and perfect recall.
\oomit{
We call the variant of ATL which is obtained by adopting these conventions 
{\em Alternating Time Logic with Knowledge
and Communicating Coalitions} and use the acronym \ATLK{} for it,
with the $iR$ underscript inherited from \cite{schobbens-atl-ir2004}
to denote the fact that the semantics utilizes strategies with perfect recall (i.e. non-memoryless)
and with incomplete information.
}

Implementing strategies which rely on distributed knowledge requires some care. 
For instance, simply supplying coalition members with a mechanism to share their observations with 
each other would have the side effect of enhancing the knowledge at each agent's 
disposal upon considering the reachability of subsequent goals as part of possibly 
different coalitions, whereas we assume that each agent's knowledge is just what 
follows from its personal experience at all times. Therefore we assume that 
coalition activities are carried out through the guidance of corresponding 
{\em virtual supervisors} who receive the coalition members' observations and 
previously accumulated knowledge and in return direct their actions for 
as long as the coalition exists without making any additional information available.

In our previous work models are based on {\em interpreted systems} as
known from \cite{FaginHalpernVardi}. 
In that setting global system states
are tuples which consist of the local views of the individual agents
and the satisfaction of atomic propositions at a global state need not
be related to the local views in it. Unlike that, in this paper we
assume that the view of each agent is described as a set of atomic
propositions which the agent can observe. States which satisfy the
same observable atomic propositions are indistinguishable to the
agent.  Observability as in interpreted systems can be simulated in
this concrete observability semantics. However, the converse does not
hold, see \cite{dima-jlc2010} for details.

We prove our model-checking result by induction on the construction
the formula to be checked, like in model-checking algorithms for ATL
or CTL, with two significant differences. Firstly, the implicit
distributed knowledge operator hidden in the coalition operator is
handled by means of a ``subset construction'' for identifying states
with indistinguishable histories, a technique used for CTLK
model-checking in \cite{dima08clima}.
Secondly, checking whether in a given set of indistinguishable states
the coalition has a strategy to achieve goal $\phi$ 
involves building a tree automaton, which can be seen as a game between the coalition 
(supervisor) and the rest of the agents.
This game resembles the two-player games with one player having
imperfect information from \cite{chatterjee-imperfect-information2006}, but also has a notable
difference: the goal of the player with imperfect information is
\emph{not fully observable}. Such a goal can be achieved \emph{at
  different times} along different yet indistinguishable runs.
Therefore, we have a bookkeeping mechanism for the time of achieving
the goal along each run. 

The tree automata we use employ only ``occurrence'' accepting
conditions: the set of states occurring along each run of the tree is
required to belong to some given set of sets of states. No Muller
conditions, i.e., no restrictions on the set of states occurring
\emph{infinitely often}, are 
involved.

The model-checking algorithm proceeds by constructing 
\emph{refinements} of the given game arena $\Gamma$,
unlike in CTL and ATL model-checking where the 
only modifications of the given arena 
are the insertion of new propositional variables (corresponding to 
subformulas of the 
formula to model-check).
\oomit{
The reason why this refinement process is necessary comes from the
necessity to represent classes of histories identically observable by 
a set of agents $A$.
This is possible only by state-splitting (with the aid of a subset construction)  
a technique also used in model-checking epistemic variants of 
CTL or LTL with perfect recall.
}
This refinement enables telling apart classes of histories which are indistinguishable to coalition members. It involves splitting states by means of a subset construction. The technique is known from model-checking epistemic extensions of CTL or LTL with perfect recall.

The setting and techniques presented here are different from 
those in our previous work \cite{guelev-dima2008dalt}.
In \cite{guelev-dima2008dalt}, the knowledge modalities are required to have only argument formulas from the past subset of $\mathit{LTL}$. 
\ATLK{} has only future operators.
Past $\mathit{LTL}$ operators can be added to \ATLK{}
in the usual way.
Also, the model-checking algorithm for \ATLK{} is based 
on tree-automata and not on the syntactical transformation of 
past formulas as in \cite{guelev-dima2008dalt}.

Let us also note the difference between our work and the work on ATEL:
the approach proposed in ATEL is to consider that strategies 
are defined on \emph{sequences of states}, which is a \emph{perfect observability} approach. 
Hence, a formula of the 
form $\ceoal[Alice] \phi$, saying that $Alice$ has a strategy to ensure $\phi$ in a given state, refers to the situation in which $Alice$ would be able to ensure $\phi$ if she had complete information about the system state. As in general agents do not have complete information, ATEL proposes then to use knowledge operators as a means to model imperfect information. The idea is to use formulas of the 
form $K_{Alice} \ceoal[Alice] \phi$ to specify the fact that $Alice$ knows that she can enforce $\phi$ in the current state.

Unfortunately, this does not solve the \emph{unfeasible strategies} problem,
studied in \cite{goranko-jamroga-comparing-2004}. Namely, the knowledge operator in formula $K_{Alice} \ceoal[Alice] \phi$
does not give $Alice$ the ability to know what action she has to apply in the current state.
This is because the knowledge operator only gives evidence about the fact that 
\emph{strategies exist, in all identically observable states, to ensure $\phi$},
but different strategies may exist in identically observable states,
and hence $Alice$ might not be able to know what strategy she is to apply after some 
sequence of observations.

Another argument against the possibility to encode the setting from 
e.g. \cite{schobbens-atl-ir2004} into the ATEL setting from \cite{wiebe-ATEL2003}
refers to the difficulty of giving a fixpoint definition to the operators involving $\ceoal[Alice]$.
The reason is that, for formulas of the form $\cefff \phi$, it is possible that $\phi$ becomes satisfied at different times along different yet indistinguishable runs. Hence, despite that $Alice$ can enforce $\phi$ by means of a fixed strategy, she might be unable to tell when $\phi$ happens. At best, in case every global state has only finitely many successors, $Alice$ would eventually be able to tell that \emph{$\phi$ must have been achieved}.
This observation is related with the bookkeeping mechanism used in Subsection \ref{sec:4.2} here, in 
the association of a tree automaton with each subformula of the form $\ceoal \phi_1 \untl \phi_2$.
\oomit{
For instance, for formulas of the type $\cefff \phi$, this difficulty stems from the fact that
the instance where $p$ is finally satisfied is not the same for identically observable runs.
This means that $Alice$ might know that, in a given state, 
she can enforce $\phi$ by playing some fixed strategy, but, along the strategy,
she might not know exactly when $\phi$ happens -- she is only able to know
that, at some point, \emph{$\phi$ must have happened in the past}.
This remark is also related with the bookkeeping mechanism used in 
the association of a tree automaton with each subformula of the type $\ceoal \phi_1 \untl \phi_2$.
}

In conclusion, we believe that there is little hope to encode the imperfect information setting
studied here within the ATEL framework from \cite{wiebe-ATEL2003,goranko-jamroga-comparing-2004}.

\noindent
{\em Structure of the paper } The next section recalls some basic
notions and notations used throughout the paper, including the
tree automata that are used in the model-checking algorithm.  Section
\ref{atldef} presents the syntax and semantics of \ATLK{}.  Section
\ref{sec3} gives the constructions involved in the model-checking
algorithm: the subset construction for identifying indistinguishable
histories, and then the tree automata for handling formulas of the
forms 
$\ceoal p_1 \wkuntl p_2$
and $\ceoal p_1\untl p_2$, respectively. 
We conclude by a summary of our result, discussion and topics of further work.
\oomit{We finish with conclusions, discussions and further work.
}

\section{Preliminaries}
\label{prelims}

Given a set $A$, $A^*$ stands for the set of finite sequences over $A$. 
The empty sequence is denoted by $\eps$. The prefix order between sequences is denoted by $\preceq$ and
the concatenation of sequences by $\cdot$.
The {\em direct product} of a family of sets $(X_a)_{a\in A}$
is denoted by $\prod_{a\in A} X_a$. 
An element $x$ 
of $\prod_{a\in A} X_a$ will be written in the
form $x=(x_a)_{a\in A}$, where $x_a\in X_a$ for all $a\in A$. 
If $B\subseteq A$, then $x\restr{B}=(x_b)_{b\in B}$ stands for the 
{\em restriction} of $x$ to $B$.
If the index set $A$ is a set of natural numbers and $n \in A$, 
then $x\restr{n}$ stands for 
$x\restr{\{n\}}$.
The \emph{support}
$\supp(f)$
of a partial function $f : A \sdpart B$ is the subset of elements of $A$ 
on which the function is defined.

Given a set of symbols $\Delta$, 
a \emph{$\Delta$-labeled tree} is a partial function 
$t : \Nset^* {\sdpart} \Delta$ 
such that
\begin{enumerate}
\item $\eps \in \supp(t)$.
\item The support of $t$ is prefix-closed: if $x \in \supp(t)$ and 
$y\preceq x$, then $y  \in \supp(t)$.
\item Trees are ``full'': if $xi \in \supp(t)$, then $xj \in \supp(t)$ for all $j\leq i$ too.
\item All tree branches are infinite:
If $x \in \supp(t)$ then $x0 \in \supp(t)$ too.
\end{enumerate}
Elements of $\supp(t)$ are called {\em nodes} of $t$.
A \emph{path} in 
$t$ is an infinite sequence of nodes $\pi = (x_k)_{k\geq 0}$ such that for all $k$,
$x_{k+1}$ is an {\em immediate} successor of $x_k$, i.e. $x_{k+1} = x_kl$ for some $l \in \Nset$.
\oomit{
The path is \emph{initialized} if it starts with the tree root, $x_0 = \eps$.
}
Path $(x_k)_{k\geq 0}$ is \emph{initialized} if $x_0$ is the tree root $\eps$.
We denote the set of labels on the path $\pi$, that is, $\{t(x_k)\mid k\geq 0\}$, by $t(\pi)$.

Below we use tree automata $\AAA = (Q, \Sigma, \delta, Q_0, \FFF)$ in which $Q$ is the set of {\em states}, 
$\Sigma$ is the \emph{alphabet}, 
$Q_0\subseteq Q$ is the set of the {\em initial} states, $\delta \subseteq Q \times \Sigma \times (2^Q\setminus \emptyset)$ is the {\em transition relation} 
and the acceptance condition $\FFF$ is a subset of $2^Q$.

Tree automata accept $Q\times \Sigma $-labelled trees.
A tree $t : \Nset ^* \sdpart Q \times \Sigma$ represents an \emph{accepting run} in $\AAA$ iff:
\begin{enumerate}
\item $t(\eps) \in Q_0 \times \Sigma$.
\item If $x \in \supp(t)$, then $t(xi)\restr{Q}\not= t(xj)\restr{Q}$ 
whenever $i\not=j$, and $(t(x)\restr{Q},t(x)\restr{\Sigma}, \{t(xi)\restr{Q}\mid xi\in \supp(t)\}) \in \delta$.
\item $t(\pi)\restr{Q}\in\FFF$ for all \emph{initialized} paths $\pi\subseteq\supp(t)$.
\end{enumerate}
Note that we only consider automata with ``occurrence'' accepting conditions: an initialized path is accepting if the 
set of states \emph{occurring} on the path is a member of $\FFF$, even if some of these states occur only finitely many times. 

\begin{theorem}[\cite{thomas97handbook}]
The emptiness problem for tree automata with ``occurrence'' accepting conditions,
i.e., the problem of checking whether, given a tree automaton $\AAA$,
there exists an accepting run in $\AAA$, is decidable.
\end{theorem}

\section{Syntax and semantics of \ATLK{}}
\label{atldef}

Throughout this paper we fix a non-empty finite set $Ag$ of \emph{agents} and, for each $a \in Ag$, 
a set of \emph{atomic propositions} $Prop_a$, 
which are assumed to be observable to $a$.
Given $A \subseteq Ag$, we write $Prop_A$ for $\bigcup_{a \in A} Prop_a$. We abbreviate $Prop_{Ag}$ to $Prop$. 

\subsection{Game arenas}

\begin{definition}
A {\em game arena} is a tuple 
$\Gamma=(Ag,Q,(C_a)_{a\in Ag},\delta,Q_0,(Prop_a)_{a\in Ag},\lambda)$,
where 
\begin{itemize}
\item $Ag$ and $Prop_a$, $a\in Ag$, are as above.
\item $Q$ is a set of {\em states}, 
\item $C_a$ is a finite sets of {\em actions} available to agent $a$.
We write $C_A$ for $\prod_{a\in A}C_a$ and $C$ for $C_{Ag}$.
\item $Q_0\subseteq Q$ is the set of {\em initial states}.
\item $\lambda : Q \sd 2^{Prop}$ is the {\em state-labeling function}.
\item 
$\delta: Q\times C\sd (2^Q\setminus \emptyset)$ 
is the {\em transition relation}.
\end{itemize}
\end{definition}
An element $c\in C$ will be called an {\em action tuple}.  We write
$q\sdup{c} r$ for {\em transitions} $(q,c,r)\in\delta$. We define
$\lambda_A:Q \sd 2^{Prop_A}$, $A\subseteq Ag$, by putting
$\lambda_A(q)=\lambda(q)\cap Prop_A$.  We assume that $\lambda$ and
$\lambda_A$ are defined on subsets $S$ of $Q$ by putting
$\lambda(S)=\bigcup\limits_{q\in S}\lambda(q)$ for $\lambda$, and
similarly for $\lambda_A$.  

Given an arena $\Gamma$, a {\em run}
$\rho$ is a sequence of transitions $q_i'\sdup{c_i} q_i''$ such that
$q_{i+1}'=q_i''$ for all $i$.  We write $\rho=(q_{i-1}\sdup{c_i}
q_i)_{1\leq i\leq n}$, resp. $\rho=(q_{i-1}\sdup{c_i} q_i)_{i\geq 1}$
for finite, resp. infinite runs. The {\em length} of 
$\rho$, denoted $|\rho|$, is the number of its transitions. This is $\infty$ for infinite
runs. A run $\rho=q_0\sdup{c_1}q_1\sdup{c_2}\ldots$ is {\em
  initialized} if $q_0\in Q_0$. $\runsfin(\Gamma)$ denotes the set of initialized finite runs and
$\runsinf(\Gamma)$ denotes the set of initialized infinite runs of $\Gamma$.

Given a run $\rho=q_0\sdup{c_1}q_1\sdup{c_2}\ldots$, we denote $q_i$
by $\rho[i]$, $i=0,\ldots,|\rho|$, and $c_{i+1}$ by $act(\rho,i)$,
$i=0,\ldots,|\rho|-1$. We write $\rho[0..i]$ for the {\em prefix}
$q_0\sdup{c_1}q_1\sdup{c_1}\ldots\sdup{c_i}q_i$ of $\rho$ of length
$i$.

A \emph{coalition} is a subset of $Ag$.
Given a coalition $A$, $S\subseteq Q$, $c_A\in C_A$, and 
$Z\subseteq Prop_A$, the following set denotes the {\em outcome of $c_A$ from $S$, labeled with $Z$}:
\[
  out(S,c_A,Z) = \{s'\in Q \mid (\exists s\in S,  \exists c'\in C) 
               c'\restr{A} = c_A, s\sdup{c'}s' \in\delta \text{ and } \lambda_A(s')=Z\}
\]
whereas those from $Prop_A\setminus Z$ are false.

Runs $\rho$ and $\rho'$ are {\em indistinguishable (observationally equivalent)} to coalition $A$, denoted $\rho \sim_A \rho'$, if $|\rho|=|\rho'|$,
$act(\rho,i)\restr{A}=act(\rho',i)\restr{A}$ for all $i<|\rho|$, and
$\lambda_A(\rho[i])=\lambda_A(\rho'[i])$ for all $i\leq |\rho|$.

\begin{definition}
A \emph{strategy} for a coalition $A$ is any mapping 
$\sigma : (2^{Prop_A})^* \sd C_A$.
\end{definition}
We write $\Sigma(A,\Gamma)$ for the set of all strategies of coalition
$A$ in game arena $\Gamma$.

Note that, instead of describing strategies for coalitions as tuples
of strategies for their individual members with every member choosing
its actions using just its own view of the past, we assume a {\em
  joint} strategy in which the actions of every coalition member
depend on the combined view of the past of all the members.
We may therefore assume that the coalition is 
guided by a supervisor who receives the members' view of the current 
state, and, in return, advices every coalition member of its next action. 
The supervisor sends no other information.
We refer the reader to a short discussion in the last section, on this supervisor interpretation 
of joint strategies.
\oomit{
We may therefore consider that the coalition is 
ruled by a supervisor, which receives 
from all the agents the information relative to the agents' view of the current 
state, and, in return, decides that some action tuple is to be issued,
and sends individual orders to each agent, requesting them to issue their action 
from the action tuple. 
On the other hand, the supervisor is not supposed to send back any supplimentary
information about the current state to any of the agents.
We refer the reader to a short discussion in the last section, on this supervisor interpretation 
of joint strategies.
}

Finite sequences of subsets of $Prop_A$ will be called \emph{$A$-histories}.

Strategy $\sigma$ for coalition $A$ is {\em compatible} 
with a run $\rho = q_0\sdup{c_1}q_1\sdup{c_2}\ldots$ if
\[
\sigma(\lambda_A(\rho[0])\cdots \lambda_A(\rho[i]))= c_{i+1}\restr{A}
\]
for all $i\leq|\rho|$. Obviously if $\sigma$ is compatible with 
run $\rho$ then it is compatible with any run that is 
indistinguishable from $\rho$ to $A$. 

\subsection{\ATLK{} defined}

The syntax of \ATLK{} formulas $\phi$ can be defined by the 
grammar
\[
\phi ::= p \mid \phi\wedge\phi  \mid \neg\phi  \mid \cenxt\phi \mid \ceoal\phi\untl\phi \mid 
\ceoal \phi_1 \wkuntl \phi_2 \mid K_A\phi
\]
where $p$ ranges over the set $Prop$ of atomic propositions,
and $A$ ranges over the set of subsets of $Ag$. 

Below it becomes clear that admitting $\wkuntl$ as a basic temporal connective allows us to introduce all the remaining combinations of $\ceoal$ and its dual $\caoal$ and the temporal connectives as syntactic sugar (see \cite{bulling-harder-2009,oreiby-2008} for more details).
\oomit{
Note that, according to \cite{bulling-harder-2009,oreiby-2008},
we include the weak-until operator in order to fully capture 
all the temporal operators inside the coalition operators.
}
Satisfaction of \ATLK{} formulas is defined with respect to a given
arena $\Gamma$, a run $\rho\in\runsinf(\Gamma)$ and a position $i$ in
$\rho$ by the clauses:
\begin{itemize}
\setlength\itemsep{1ex}
\item $(\Gamma,\rho,i)\models p$ if $p\in \lambda(\rho[i])$.
\item $(\Gamma,\rho,i) \models \phi_1 \wedge \phi_2$,  
      if $(\Gamma,\rho,i)\models \phi_1$ and 
      $(\Gamma,\rho,i)\models\phi_2$.
\item $(\Gamma,\rho,i) \models \neg \phi$ if
      $(\Gamma,\rho,i) \not \models \phi$.
\item $(\Gamma,\rho,i) \models \cenxt\phi$ if 
      there exists a strategy $\sigma\in \Sigma(A,\Gamma)$ such that 
      $(\Gamma,\rho',i+1) \models \phi$ for all runs
      $\rho'\in\runsinf(\Gamma)$ which are compatible with $\sigma$
      and satisfy $\rho'[0..i] \sim_A \rho[0..i]$.
\item $(\Gamma,\rho,i) \models \ceoal\phi_1\untl\phi_2$ iff
      there exists a strategy $\sigma\in \Sigma(A,\Gamma)$  such that 
      for every run $\rho'\in\runsinf(\Gamma)$ which is compatible with $\sigma$
      and satisfies $\rho'[0..i] \sim_A \rho[0..i]$ there exists $j\geq i$
      such that $(\Gamma,\rho',j) \models \phi_2$ and  
      $(\Gamma,\rho',k) \models \phi_1$ for all $k=i,\ldots,j-1$.
\item $(\Gamma,\rho,i) \models \ceoal\phi_1\wkuntl\phi_2$ iff
      there exists a strategy $\sigma\in \Sigma(A,\Gamma)$  such that 
      for every run $\rho'\in\runsinf(\Gamma)$ which is compatible with $\sigma$
      and satisfies $\rho'[0..i] \sim_A \rho[0..i]$ one of the two situations occur: 
\begin{enumerate}
        \item Either there exists $j\geq i$
      such that $(\Gamma,\rho',j) \models \phi_2$ and  
      $(\Gamma,\rho',k) \models \phi_1$ for all $k=i,\ldots,j-1$.
\item Or $(\Gamma,\rho',k) \models \phi_1$ for all $k \geq i$.
\end{enumerate}
\item $(\Gamma,\rho,i) \models K_A \phi$ iff 
      $(\Gamma,\rho',i) \models \phi$, for all runs 
      $\rho'\in\runsinf(\Gamma)$ 
      which satisfy $\rho'[0..i] \sim_A \rho[0..i]$.
\end{itemize}

The rest of the combinations between the temporal connectives and the cooperation modalities $\ceoal$ and $\caoal$ are defined as follows:
\begin{align*}
P_A \phi & = \neg K_A \neg \phi & \canxt\phi& = \neg\cenxt\neg\phi \\
\caoal\phi\untl\psi &= \neg\ceoal(\neg\psi\wkuntl\neg\psi\wedge\neg\varphi) &
\caoal\phi\wkuntl\psi &= \neg\ceoal(\neg\psi\untl\neg\psi\wedge\neg\varphi)\\
\cefff \phi & = \ceoal \true \untl \phi & 
\ceggg \phi & = \ceoal \phi \wkuntl \false \\
\cafff \phi & = \caoal\true\untl\phi & 
\caggg \phi & = \caoal\phi\wkuntl\false 
\end{align*}
\oomit{The following abbreviations are common to all variants of ATL:
\begin{align*}
\cefff \phi & = \ceoal \true \untl \phi & 
\ceggg \phi & = \ceoal \phi \wkuntl \false \\
\cafff \phi & = \neg \ceggg \neg \phi & 
\caggg \phi & = \neg \cefff \neg \phi \\ 
P_A \phi & = \neg K_A \neg \phi 
\end{align*}
}

A formula $\phi$ is \emph{valid in a game arena $\Gamma$}, written
$\Gamma \models \phi$, if $(\Gamma,\rho,0)\models\phi$ for all
$\rho\in\runsinf(\Gamma)$. The \emph{model-checking problem} for
\ATLK{} is to decide whether $\Gamma \models \phi$ for a given formula
$\phi$ and arena $\Gamma$.

\oomit{
\begin{remark}
Consider a tree $t$ over $Q\times C$ which satisfies the following properties:
\begin{quote}
For each $x \in \supp(t)$ and $c \in C$ there exists $i \in \Nset$ with 
$xi \in \supp(t)$ and $q' \in Q$ such that $t(xi) = (q',c)$.
\end{quote}
Such a tree can be regarded as a game arena: its nodes are the states of the game arena,
and the relation between a node and its sons gives the transition relation.

Conversely, to each game arena $\Gamma$ one may associate a tree $t : \Nset \sdpart Q\times C$,
which is the usual \emph{unfolding} of $\Gamma$:
nodes are labeled with tuples consisting of a state in $Q$ and a ``guess'' of the 
actions that each agent will play in that state, and 
the set of transitions originating in a state $q$ and with a tuple of 
actions $c \in C$ gives rise to the 
relation between a node labeled with $(q,c)$ and its successors,
successors, which are labeled with all combinations $(r,c')$ with $(q,c,r) \in\delta $ and $c'\in C$ arbitrary.
Actually one may associate many trees, all of which are similar. 
\end{remark}
\begin{proposition}
For each game arena $\Gamma$ and each \ATLK{} formula $\phi$, 
$\Gamma \models \phi$ if and only if for any tree unfolding $t$ of $\Gamma$ we have that 
$t\models \phi$. 
\end{proposition}
}

\begin{example}\label{ex:ex}

 Alice and Bob, married, work in the same company.
  When they  arrive at work, 
	they are assigned (by some non-modeled agent) 
  one of the tasks $x$ or $y$. These
  tasks need different periods of time to be executed: $tx$ time units
  for $x$ and $ty$ time units for $y$, where $tx<ty$. 
	The assignment is always such that 
  task $y$ cannot be assigned to both Alice and Bob. After they finished
  executing their task, Alice and Bob have two objectives: (1) to pick
  their child from the nursery, and (2) to do the shopping. 
  The supermarket closes early, so the one who does the longest task
  cannot do the shopping.
  So Alice and Bob need to exchange information about their 
  assigned task in order to fix who's to do the shopping and 
  who's to pick the child from the nursery.
   
  Figure \ref{fig:ex} pictures the game arena representing this
  system. 
	The actions 
  Alice and Bob can do are: $g$ for going at work,
  $e$ for working on their task, $tc$ for taking
  the child, $ds$ for doing the shopping, and $i$ for idling. 
  The atomic proposition
  $xx$ denotes the assignment of task $x$ to both Alice and Bob, $xy$
  denotes assignment of task $x$ to Alice and task $y$ to Bob, and $yx$
  denotes assignment of task $y$ to Alice and task $x$ to Bob. All these atomic propositions 
  are not observable by the two agents. The atomic
  propositions $x_a$ and $y_a$ are
  observed only by Alice, and the atomic propositions $x_b$ and $y_b$ are observed only by Bob. 
  All these four atoms denote the fact
  that the respective person has to execute task $x$ or $y$. Furthermore, the atomic
  propositions $tx_a$ and $ty_a$ which  are
  observed only by Alice, and $tx_b$ and $ty_b$, which are observed only by Bob, 
  denote the fact that the respective person has been working for $tx$ or $ty$ time units. The atomic
  propositions $c$, $s$, can be observed by both Alice and Bob and
  denote the fact that the child was picked from the nursery, 
  and, respectively, that the fridge is full with food from the supermarket. 
An arc labeled by two vectors of actions, e.g.
  $(tc,ds)\ (ds,tc)$, denotes two arcs with the same origin and the
  same destination, each one of them labeled by one of the vectors.
  
  We suppose that the game arena contains a {\em sink state} which is the output of all the transitions 
  not pictured in Figure \ref{fig:ex} (for instance, both agents idling in state $q_6$ brings the sistem to the sink state). 
  Also, we suppose that all the states except for the sink state are labeled by some atomic proposition $valid$ visible to Alice.
  
  An interesting property for this system is that Alice and Bob can
  form a coalition in order to pick their child and do
  their shopping (if we ignore the sink state)-- that is, the following formula is true:
\[
\phi= \ceoal[\{Alice,Bob\}] (valid\ \untl\ c\wedge s)
\]

Note that Alice and Bob need a strategy which must include some 
communication during its execution, which would help each of them to know 
who is assigned which task during the day, and hence who cannot do the shopping.
Note also that the model incorporates some timing information,
such that the two agents need a strategy with perfect recall in order
to reach their goal: after working $tx$ time units both Alice and Bob must use their observable past
to remember if they have finished working.
Finally, note that, if we consider that strategies for
coalitions are tuples of strategies for individual members, as in 
\cite{alur-atl1997,schobbens-atl-ir2004} then
the formula $\phi$ is false: 
whatever decision Alice and Bob take together, in the morning,
about who is to pick the child, who is to do shopping,
and in what observable circumstances (but without exchanging any information), 
can be countered by the task assignment,
which would bring Alice and Bob at the end of the day either with an empty fridge
or  the child spending his night at the nursery.
\end{example}

\setlength{\unitlength}{1.7pt}

\vspace{-0.7cm}
\begin{figure}[ht]
\begin{center}
\hspace{-3cm}
\begin{center}
\begin{picture}(264,131)(0,-131)
\node[Nw=18.7,Nh=14.82,Nmr=7.41](n1)(12.0,-64.0){}

\node[NLangle=0.0,Nw=18.7,Nh=14.82,Nmr=7.41](n11)(44.0,-36.0){{\small $yx$}}

\node[NLangle=0.0,Nw=18.7,Nh=14.82,Nmr=7.41](n12)(44.0,-64.0){{\small $xx$}}

\node[NLangle=0.0,Nw=18.7,Nh=14.82,Nmr=7.41](n13)(44.0,-92.0){{\small $xy$}}

\drawedge[linewidth=1.0,ELdist=0.0,AHangle=29.62,AHLength=2.93](n1,n11){{\footnotesize $(g,g)$}}

\drawedge[linewidth=1.0,ELdist=0.0,AHangle=32.62,AHLength=2.97](n1,n12){{\footnotesize $(g,g)$}}



\drawedge[linewidth=1.0,ELpos=60,ELdist=0.0,AHangle=28.44,AHLength=2.73](n1,n13){{\footnotesize $(g,g)$}}


\node[NLangle=0.0,Nw=18.7,Nh=14.82,Nmr=7.41](n14)(151.86,-36.0){{\small $ty_a,tx_b$}}

\node[NLangle=0.0,Nw=18.7,Nh=14.82,Nmr=7.41](n15)(151.86,-64.0){{\small $tx_a,tx_b$}}

\node[NLangle=0.0,Nw=18.7,Nh=14.82,Nmr=7.41](n16)(151.86,-92.0){{\small $tx_a,ty_b$}}

\node[NLangle=0.0,Nw=18.7,Nh=14.82,Nmr=7.41](n17)(231.86,-20.0){$c,s$}

\node[NLangle=0.0,Nw=18.7,Nh=14.82,Nmr=7.41](n18)(243.62,-64.0){$c$}

\node[NLangle=0.0,Nw=18.7,Nh=14.82,Nmr=7.41](n19)(235.73,-107.87){$s$}

\drawedge[linewidth=1.0,ELdist=0.0,AHangle=27.47,AHLength=2.82](n14,n17){{\footnotesize $(tc,ds)$}}

\drawedge[linewidth=1.0,ELpos=85,ELdist=0.0,AHangle=26.57,AHLength=2.68](n14,n18){{\footnotesize $(tc,tc)(ds,tc)$}}

\drawedge[linewidth=1.0,ELpos=85,ELdist=0.0,AHangle=27.03,AHLength=2.75](n14,n19){{\footnotesize $(ds,ds)$}}

\drawedge[linewidth=1.0,ELpos=60,ELdist=0.0,AHangle=26.57,AHLength=2.57](n15,n18){{\footnotesize $(tc,tc)$}}

\drawedge[linewidth=1.0,ELpos=55,ELdist=-12.0,AHangle=23.86,AHLength=2.84](n15,n19){{\footnotesize $(ds,ds)$}}

\drawedge[linewidth=1.0,ELpos=80,ELdist=-13.0,AHangle=27.05,AHLength=2.64](n16,n17){{\footnotesize $(ds,tc)$}}

\drawedge[linewidth=1.0,ELpos=80,ELdist=-14.0,AHangle=25.2,AHLength=2.82](n16,n18){{\footnotesize $(tc,tc) (tc,ds)$}}

\drawedge[linewidth=1.0,ELdist=-10.0,AHangle=22.78,AHLength=2.71](n16,n19){{\footnotesize $(ds,ds)$}}

\drawedge[linewidth=1.0,ELpos=45,ELdist=0.0,AHangle=30.47,AHLength=2.96](n15,n17){{\footnotesize $(tc,ds)(ds,tc)$}}

\node[NLangle=0.0,Nw=18.7,Nh=14.82,Nmr=7.41](n29)(76.0,-36.0){{\small $tx_a,tx_b$}}

\node[NLangle=0.0,Nw=18.7,Nh=14.82,Nmr=7.41](n30)(76.0,-92.0){{\small $tx_a,tx_b$}}

\drawedge[linewidth=1.0,ELdist=0.0,AHangle=26.05,AHLength=2.5](n11,n29){{\footnotesize $(i,i)$}}

\drawedge[linewidth=1.0,ELdist=0.0,AHangle=22.78,AHLength=2.71](n13,n30){{\footnotesize $(i,i)$}}



\node[NLangle=0.0,Nw=18.7,Nh=14.82,Nmr=7.41](n33)(76.0,-64.0){{\small $tx_a,tx_b$}}

\drawedge[linewidth=1.0,ELdist=0.0,AHangle=27.05,AHLength=2.64](n12,n33){{\footnotesize $(i,i)$}}


\node[NLangle=0.0,Nw=18.7,Nh=14.82,Nmr=7.41](n35)(116.0,-36.0){{\small $tx_a,tx_b$}}

\node[NLangle=0.0,Nw=18.7,Nh=14.82,Nmr=7.41](n36)(116.0,-92.0){{\small $tx_a,tx_b$}}

\drawedge[linewidth=1.0,ELdist=0.0,AHangle=28.95,AHLength=2.69](n29,n35){{\footnotesize $(e,e)$}}

\drawedge[linewidth=1.0,ELdist=0.0,AHangle=26.57,AHLength=2.8](n33,n15){{\footnotesize $(e,e)$}}

\drawedge[linewidth=1.0,ELdist=0.0,AHangle=26.57,AHLength=2.68](n30,n36){{\footnotesize $(e,e)$}}

\drawedge[linewidth=1.0,ELdist=0.0,AHangle=25.64,AHLength=2.77](n35,n14){{\footnotesize $(e,i)$}}

\drawedge[linewidth=1.0,ELdist=0.0,AHangle=25.6,AHLength=2.66](n36,n16){{\footnotesize $(i,e)$}}


\node[Nframe=n,NLangle=0.0,Nw=9.52,Nh=7.76,Nmr=0.0](n41)(44.0,-24.0){$q_1$}

\node[Nframe=n,NLangle=0.0,Nw=9.52,Nh=7.76,Nmr=0.0](n42)(44.0,-52.0){$q_2$}

\node[Nframe=n,NLangle=0.0,Nw=9.52,Nh=7.76,Nmr=0.0](n43)(44.0,-104.0){$q_3$}

\node[Nframe=n,NLangle=0.0,Nw=9.52,Nh=7.76,Nmr=0.0](n20)(76.0,-24.0){$q_4$}

\node[Nframe=n,NLangle=0.0,Nw=9.52,Nh=7.76,Nmr=0.0](n21)(76.0,-52.0){$q_5$}

\node[Nframe=n,NLangle=0.0,Nw=9.52,Nh=7.76,Nmr=0.0](n22)(76.0,-104.0){$q_6$}

\node[Nframe=n,NLangle=0.0,Nw=9.52,Nh=7.76,Nmr=0.0](n23)(116.0,-24.0){$q_7$}

\node[Nframe=n,NLangle=0.0,Nw=9.52,Nh=7.76,Nmr=0.0](n24)(152.0,-24.0){$q_9$}

\node[Nframe=n,NLangle=0.0,Nw=9.52,Nh=7.76,Nmr=0.0](n25)(152.0,-52.0){$q_{10}$}

\node[Nframe=n,NLangle=0.0,Nw=9.52,Nh=7.76,Nmr=0.0](n26)(116.0,-104.0){$q_8$}

\node[Nframe=n,NLangle=0.0,Nw=9.52,Nh=7.76,Nmr=0.0](n27)(152.0,-104.0){$q_{11}$}

\node[Nframe=n,NLangle=0.0,Nw=9.52,Nh=7.76,Nmr=0.0](n28)(232.0,-8.0){$q_{12}$}

\node[Nframe=n,NLangle=0.0,Nw=9.52,Nh=7.76,Nmr=0.0](n29)(252.0,-74.0){$q_{13}$}

\node[Nframe=n,NLangle=0.0,Nw=9.52,Nh=7.76,Nmr=0.0](n30)(236.0,-120.0){$q_{14}$}

\drawloop[linewidth=1.0,AHangle=23.5,AHLength=2.51,loopdiam=10.49,loopangle=151.14](n17){{\footnotesize $(i,i)$}}

\drawloop[linewidth=1.0,AHangle=23.5,AHLength=2.51,loopdiam=9.48,loopangle=40](n18){{\footnotesize $(i,i)$}}

\drawloop[linewidth=1.0,AHangle=23.5,AHLength=2.51,loopdiam=9.48,loopangle=55](n19){{\footnotesize $(i,i)$}}
\end{picture}
\end{center}
\end{center}
\vspace{-0.8cm}
\caption{A game arena for Example \ref{ex:ex}}
\label{fig:ex}
\end{figure}
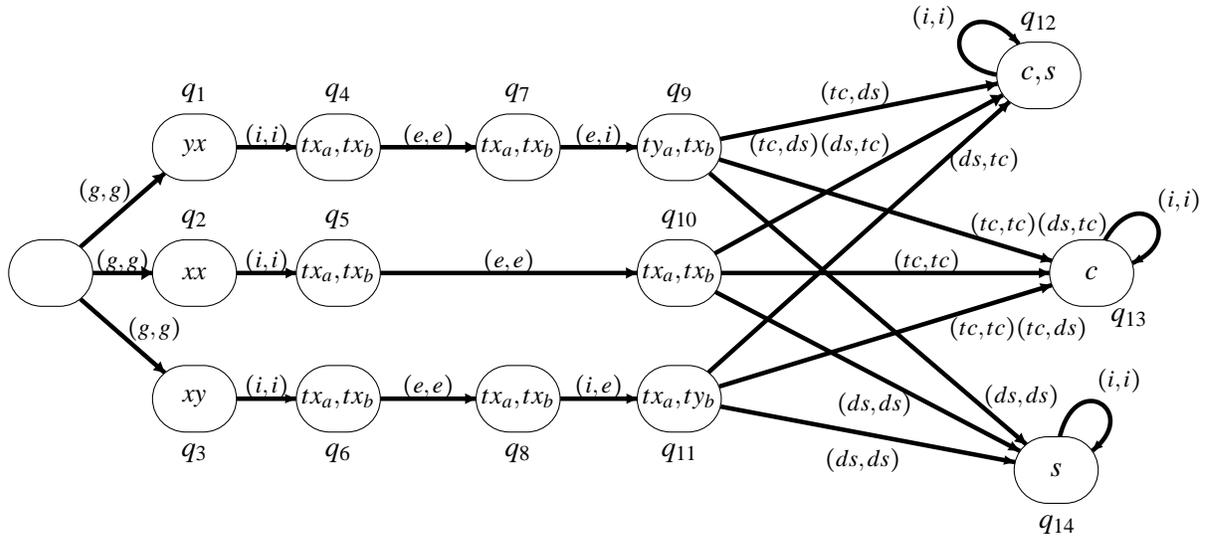

\section{Model-checking \ATLK{}}
\label{sec3}

The model-checking procedure for \ATLK{} builds on model checking
techniques for CTL with knowledge modalities and ATL with complete
information. It works by recursion on the construction of formulas.
Given a formula $\phi$ with a cooperation modality as the main
connective, the procedure involves refining the given arena $\Gamma$
to an arena $\widehat \Gamma$ in which the state space can be
partitioned into states which satisfy $\phi$ and states which do not.

The idea is to have, after the splitting, an equivalence relation $\equiv_A$ on the states of the resulting game arena $\widehat \Gamma$,
such that 
$\widehat {q_1} \equiv_A \widehat {q_2}$ 
iff $\widehat{q_1}$ and $\widehat{q_2}$ are reachable through the same histories, as seen by $A$.

The construction 
of the refined state space
is inspired 
by the usual construction of a game with perfect information for solving 
two-player games with one player having imperfect information, see \cite{chatterjee-imperfect-information2006}.
However the construction is more involved, because, contrary to \cite{chatterjee-imperfect-information2006}
the objectives here may not be observable by the coalition.

\subsection{The state-splitting construction.}
\label{sec:031}
Given a game arena 
$\Gamma=(Ag,Q,(C_a)_{a\in Ag},\delta,Q_0,(Prop_a)_{a\in Ag},\lambda)$ 
and a coalition $A$, we construct a new game arena 
$\widehat {\Gamma_A}=(Ag,\widehat Q, (C_a)_{a \in Ag}, \widehat \delta, 
              \widehat {Q_0}, (Prop_a)_{a\in Ag}, \widehat \lambda)$,
as follows:
\begin{itemize}
\setlength\itemsep{.5ex}
\item $\widehat Q =\big\{(q,S)\mid S\subseteq Q,  q\in S\text{ and for all }
                     s\in S, \lambda_A(s)=\lambda_A(q) \big\}$;
\item $\widehat {Q_0}=\big\{(q_0,S_0) \mid q_0 \in Q_0 \text{ and } 
                     S_0=\{s\in Q_0 \mid \lambda_A(s)=\lambda_A(q_0)\}\big\}$;
\item $\widehat \lambda(q,S) = \lambda(q)$ for all $(q,S) \in \widehat Q$.
\item $(q,S) \sdup{c} (q',S')\in\widehat\delta$ if and only if
      the following properties hold:
  \setlength\parskip{1ex}
  \begin{itemize}
  \setlength\itemsep{.5ex}
  \item $(q,S),(q',S')\in \widehat Q$ and $c\in C$;
  \item $q\sdup{c} q'\in \delta$;
  \item $S'=  out\big(S,c\restr{A},\lambda_A(q')\big)$.
  \end{itemize}
\end{itemize}

The intended equivalence on states is then the following:
$\widehat q \equiv_A \widehat {q'}$ if and only if there exists $S \subseteq Q$ with 
$\widehat q = (q,S)$ and $\widehat {q'} = (q',S)$.

Every run $\rho\in\runsinf(\Gamma)$, $\rho=(q_{i-1}\sdup{c_i} q_i)_{i\geq 1}$, 
has a unique corresponding run $\widehat\rho\in\runsinf(\widehat{\Gamma_A})$,
$\widehat\rho=((q_{i-1},S_{i-1})\sdup{c_i}(q_i,S_i))_{i\geq 1}$. 
This is because $q_0$ unambiguously determines $S_0$ and, 
recursively, $S_{i-1}$ uniquely determines $S_i$, for any 
$i\geq 1$. 
The converse holds  too, that is, to each run $\br \rho=((q_{i-1},S_{i-1})\sdup{c_i}(q_i,S_i))_{i\geq 1}$ in $\runsinf(\widehat {\Gamma_A})$, 
corresponds a unique run $\rho=(q_{i-1}\sdup{c_i} q_i)_{i\geq 1}$ such that $\widehat\rho=\br \rho$.
Furthemore, every strategy for $A$ in $\Gamma$ is also a strategy for $A$ in $\widehat {\Gamma_A}$.
\begin{proposition}\label{P1}
\begin{enumerate}
\item 
\label{P1item1}
If $\rho$ and $\rho'$ are runs in $\Gamma$ of the same length, then
$\rho \sim_A \rho'$ iff $\widehat\rho \sim_A \widehat {\rho'}$.
\item 
\label{P1item2}
If $B\subseteq Ag$, $\sigma\in \Sigma(B,\Gamma) = \Sigma(B,\widehat{\Gamma_A})$, and
$\rho \in \runsinf(\Gamma)$, then $\sigma$ is compatible with $\rho$ iff $\sigma$ is compatible with $\widehat\rho$. 
\item 
\label{P1item3}
If $\rho \in \runsinf(\Gamma)$,
$p\in Prop$ and $i\geq 0$, then
$(\Gamma,\rho,i)\models K_A p$ is equivalent to both $(\widehat\Gamma,\widehat\rho,i)\models K_A p$,
and to $p\in\lambda(s)$ for all $s$ in the second component of $\widehat\rho[i]$.
\item 
\label{P1item4}
If $\rho \in \runsinf(\Gamma)$, $\phi$ is an arbitrary \ATLK{} formula and $i\geq 0$, then 
\[
(\Gamma,\rho,i)\models\phi \text{ iff } (\widehat{\Gamma_A},\widehat\rho,i)\models\phi
\]
\end{enumerate}
\end{proposition}
\begin{proof}
(\ref{P1item1}), (\ref{P1item2}) and (\ref{P1item3}) follow directly from definition. 
(\ref{P1item4}) is proved by structural induction on $\phi$. For example,
\begin{itemize}
\item if $\phi=K_B \psi$, for some $B\subseteq Ag$, then $(\widehat{\Gamma_A},\widehat\rho,0)\models\phi$ iff $(\widehat{\Gamma_A},\widehat{\rho'},0) \models \psi$ for all $\widehat{\rho'}\in\runsinf(\widehat{\Gamma_A})$ such that $\lambda_B(\widehat{\rho'}[0])=\lambda_B(\widehat\rho[0])$. By the induction hypothesis, this is equivalent to $(\Gamma,\rho',0) \models \psi$ for all $\rho'\in\runsinf(\Gamma)$ such that $\lambda_B(\rho'[0])=\lambda_B(\rho[0])$. The latter is equivalent to $(\Gamma,\rho,0) \models \phi$.

\item if $\phi=\ceoal[B]\psi_1\untl\psi_2$, for some $B\subseteq Ag$, then $(\widehat{\Gamma_A},\widehat\rho,i)\models\phi$ iff there exists a strategy $\sigma\in \Sigma(B,\widehat{\Gamma})$ such that 
for every run $\widehat{\rho'}\in\runsinf(\widehat{\Gamma})$ which is compatible with $\sigma$
      and satisfies $\widehat{\rho'}[0..i] \sim_A \widehat{\rho}[0..i]$ there exists $j\geq i$
      such that $(\widehat{\Gamma},\widehat{\rho'},j) \models \psi_2$ and  
      $(\widehat{\Gamma},\widehat{\rho'},k) \models \psi_1$ for all $k=i,\ldots,j-1$. Let $\rho''\in\runsinf(\Gamma)$ be a run compatible with $\sigma$ such that $\rho''[0..i] \sim_A \rho[0..i]$.  We have that $\widehat{\rho''}[0..i]\sim_A \rho[0..i]\sim_A \widehat{\rho}[0..i]$ and by (\ref{P1item2}), $\widehat{\rho''}$ is compatible with $\sigma$. Consequently, there exists $j\geq i$
      such that $(\widehat{\Gamma},\widehat{\rho''},j) \models \psi_2$ and  
      $(\widehat{\Gamma},\widehat{\rho''},k) \models \psi_1$ for all $k=i,\ldots,j-1$. By the induction hypothesis, we obtain that $(\Gamma,\rho'',j) \models \psi_2$ and  
      $(\Gamma,\rho'',k) \models \psi_1$ for all $k=i,\ldots,j-1$ which implies $(\Gamma,\rho,i)\models\phi$. For the other implication we can proceed in a similar manner.
\end{itemize}
\end{proof}

\begin{remark}
Item (\ref{P1item3}) from Proposition \ref{P1} gives the state partitioning 
procedure for knowledge operators:
we may partition the state space of $\widehat {\Gamma_A}$ as 
$\widehat Q = \widehat Q^{K_A p} \cup \widehat Q^{\neg K_A p}$,
where 
\begin{align}
 \widehat Q^{K_A p} & = \{ (q,S) \in \widehat Q \mid 
           (\forall s\in S)(p \in \lambda(s)= \lambda(q))\} \label{id:q-hat} \\
  \widehat Q^{\neg K_A p} & = Q \setminus  \widehat Q^{K_A p} \label{id:q-neg-hat}
\end{align}

\end{remark}

\begin{example}\label{ex:ga}
The arena $\widehat{\Gamma_{\{Alice,Bob\}}}$ corresponding to $\Gamma$ from Figure \ref{fig:ex} is obtained by replacing each state $q$ with:
\begin{itemize}
        \item $(q,\{q\})$, if $q\not\in\{q_1,q_2,q_3\}$,
        \item $(q,\{q_1,q_2,q_3\})$, otherwise.
\end{itemize}

The states $(q,\{q_1,q_2,q_3\})$ with $q\in\{q_1,q_2,q_3\}$ denote the fact that, from the point of view of Alice and Bob, $q$ is reachable through the same history as the states $q_1$, $q_2$, and $q_3$.

%
\end{example}


\subsection{The state labeling constructions}

\label{sec:4.2}

Our next step is to describe how, given an arena $\Gamma$ and a
coalition $A$, the states $(q,S) \in \widehat Q$ of
$\widehat{\Gamma}_A$ can be labelled with the 
\ATLK{}
formulas which
they satisfy in case the considered formulas have one of 
the
forms $\cenxt p$, 
$\ceoal p_1 \untl p_2$ and $\ceoal p_1 \wkuntl p_2$.

The three cases are different. Formulas of the form $\cenxt p$ are the
simplest to handle. To do formulas of the forms $\ceoal
p_1 \untl p_2$ and $\ceoal p_1 \wkuntl p_2$, we build appropriate tree automata.
 
\paragraph{Case $\cenxt p$:}
We partition the state space of $\widehat {\Gamma_A}$ in 
$\widehat Q^{\cenxt p}$ and 
$\widehat Q^{\neg \cenxt p}$, where 
\begin{align}
\widehat Q^{\cenxt p} & = 
\big\{ (q,S) \in \widehat Q \mid \exists c\in C_A \text{ s.t. }
\forall S' \subseteq Q,\ \forall r \in S,\ \forall r' \in S',\ \forall c' \in C, \notag \\
 & \qquad 
\text{ if } (r,S) \sdup{c'} (r',S')\text{ and } c'\restr{A} = c \text{ then } p \in \widehat\lambda(r')
\big\} \label{id:next}\\
\widehat Q^{\neg \cenxt p} & = \widehat Q \setminus \widehat Q^{\cenxt p}
\end{align}

\paragraph{Case $\ceoal p_1 \untl p_2$:} 
We build a tree automaton 
whose states represent histories which are indistinguishable to
$A$ in a finitary way. A special mechanism  is needed for checking
whether the objective $p_1 \untl p_2$ is satisfied on all paths of an
accepted tree. The main difficulty lies in the fact that the objective 
need not be observable by coalition $A$
because neither $p_1$ nor $p_2$ are required to belong to $Prop_A$.
Hence there can be behaviours $\rho$ and $\rho'$ such that
$\rho'[0..i]\sim_A \rho[0..i]$ and $(\rho,i)$ satisfies $p_1 \untl
p_2$ but $(\rho',i)$ does not.

Therefore, given some group of states $R$ representing some history,
we need to keep track of the subset $R'$ of states in $R$ for which the
obligation $p_1 \untl p_2$ was not yet satisfied on their history.
All the states in $R'$ must be labeled with $p_1$, and we need to find
outgoing transitions in the automaton that ensure the obligation to
have  $p_1 \untl p_2$ on all paths leaving $R'$.  On the other hand, states in
$R\setminus R'$ are assumed to have histories in which $p_1 \untl p_2$
has been `` achieved'' in the past, and, therefore, are ``free'' from
the obligation to fulfill $p_1 \untl p_2$.

Let $(q,S)\in\widehat{Q}$. Formally, the tree automaton is 
$\breve \AAA_{(q,S)} = (\breve Q, C_A, \breve \delta, \breve {Q}_0, \breve \FFF)$ 
where:

\begin{itemize}
\item $\breve Q$ contains $\bot$, assumed to signal failure to fulfil $p_1\untl p_2$, 
and all the sets of pairs $(R_1,R_2)$ with:
  \begin{itemize}
  \item $R_1 \subseteq R_2 \subseteq Q$,
  \item $\forall r_1,r_2 \in R_2, \lambda_A(r_1) = \lambda_A(r_2)$,
    and $\forall r_1\in R_1, p_2\not\in \lambda(r_1)\wedge p_1\in \lambda(r_1)$,
  \end{itemize}
\item The initial state $\breve Q_0$ is defined by:
  \begin{enumerate}
  \item if there exists $ s\in S$ for which $ \lambda(s)\cap \{p_1,p_2\}=\emptyset$ then $\breve Q_0 =\bot$.
  \item otherwise, we denote $Q[p_2] = \big\{ q \in Q \mid p_2 \in \lambda(q) \big\}$ and we put
    $Q_0=\{(S\setminus Q[p_2],S) \}$.
  \end{enumerate}
\item $\breve \delta : \breve Q \times C_A \sd 2^{\breve Q} \setminus \emptyset$
 is defined as follows: first, for any $c_A\in C_A$, $\delta((\bot,c_A))=\{\bot\}$. 
Then, for each $(R_1,R_2)\in \breve Q\setminus \{\bot\}$ and $c_A \in C_A$, two situations may occur:
\begin{enumerate}
\item If there exist $r_1 \in R_1$, $(r,R)\in \widehat Q$ and $c \in C$ such that 
$(r_1,R_2)\sdup{c}(r,R) \in \widehat\delta$, $c\restr{A} = c_A $
and $\{p_1,p_2\} \cap \lambda(r)=\emptyset$, then 
$\delta\big((R_1,R_2),c_A\big) = \{\bot\}$.
\item Otherwise, 
\begin{align*}
\!\!\!\!\!\!\!\!
\!\!\!\!\!\!\!\!
\!\!\!\!\!\!\!\!
\delta\big((R_1,R_2),c_A\big) & = \big\{(out(R_1,c_A,Z) \setminus Q[p_2],out(R_2,c_A,Z)) \mid Z\subseteq Prop_A, 
 out(R_2,c_A,Z)\neq \emptyset  \big\} 
\end{align*}
That is, 
each transition from $(R_1,R_2)$ labeled with $c_A$ must embody sets of states representing 
all the variants of observations which occur as outcomes of the action tuple $c_A$ from $R_2$,
paired with the subset of states in which the $p_1 \untl p_2$ obligation is not fulfiled.
\end{enumerate}

\item The acceptance condition is 
\[
\breve \FFF = \big\{ \RRR \mid \RRR \subseteq \breve Q \text{ with } 
(\emptyset,R) \in \RRR, \text{ for some  } R \subseteq Q \big\}.
\]
That is, $\breve \AAA_{\widehat q}$ accepts only trees 
in which each path reaches some node containing the empty set as first state label.
\end{itemize}

Note that, in a pair $(R_1,R_2)$ representing an element in $\breve Q$,
the first component $R_1$ represents the subset of states  of $R_2$ whose history has not yet accomplished $p_1 \untl p_2$.
Hence, a tree node with label $(\emptyset, R)$ 
signals that the obligation $p_1 \untl p_2$ is accomplished for all histories ending in $R$.

Note also that, whenever the successors of $(R_1,R_2)$ labeled $c_A$ do not contain a state labeled by $\bot$,
we have that, for any $Z \subseteq Prop_A$ and any $s \in out(R_2,c_A,Z)$, 
$p_1 \in \lambda(s)$ or $p_2 \in \lambda(s)$.

We may then prove the following result:
\begin{proposition}\label{PU}
For any run $\widehat \rho \in \runsinf(\widehat{\Gamma_A})$ and position $i$ on the run for which $\rho[i] = \widehat q = (q,S)$,
\[
(\widehat{\Gamma_A},\rho,i) \models \ceoal p_1 \untl p_2 \text{ if and only if } 
L(\breve\AAA_{\widehat q}) \neq \emptyset
\]
\end{proposition}
\begin{proof}
($\Rightarrow$) Suppose that $(\widehat {\Gamma_A}, \rho, i) \models
\ceoal p_1 \untl p_2 $. Then, there exists $\sigma \in
\Sigma(A,\widehat {\Gamma_A})$ such that for any $\rho' \in
\runsinf(\widehat {\Gamma_A})$ compatible with $\sigma$ and for
which $\rho'[0..i] \sim_A \rho[0..i]$ we have $(\widehat
{\Gamma_A},\rho',i) \models p_1 \untl p_2$.

Let $t:\Nset^* \sdpart \breve Q\times C_A$ be a tree constructed recursively as follows:
\begin{itemize}
\item The root of the tree is $t(\epsilon)=((S\setminus Q[p_2],S),c)\in \breve Q_0$ 
where $c=\sigma(\lambda_A(\rho[0])\ldots \lambda_A(\rho[i]))$. 
Note that, by  hypothesis, $\bot\not\in \breve Q_0$.
\item Suppose we have build the tree up to level $j\geq 0$.  Let
  $t(x)=((R_1,R_2),c_A)$ be a node on the $j$th level, where $x\in
  \supp(t)\cap \Nset^j$. 
Consider some order on the set 
$\breve\delta(t(x)) = \breve\delta((R_1,R_2),c_A) = \{ (R_1^1,R_2^1),\ldots,(R_1^k,R_2^k)\}$ for some $k\geq 1$.
The successors of $t(x)$ will be labeled with the elements of this set, 
each one in pair with an action symbol in $C_A$ -- action symbol which is chosen as follows:

Denote $(x_p)_{1\leq p\leq j}$ the initialized path in $t$ which ends in $x$.
For each $1\leq l\leq k$, put 
\[
c_l = \sigma(\lambda_A(t(x_1)\restr{\breve Q})\ldots \lambda_A(t(x_k)\restr{\breve Q})\lambda_A(R_2^l)).
\]
Then, for all $1\leq l\leq k$ we put $t(xl)=((R_1^l,R_2^l),c_l)$.

\end{itemize}

Suppose that $L(\breve \AAA_{\widehat q})=\emptyset$. This implies
that $t$ is not an accepting run in $\AAA$. Consequently, there exists an
infinite path $\pi=(x_k)_{k\geq 0}$, where $x_k\in\Nset^k$, in $t$
which does not satisfy any acceptance condition in $\breve\FFF$. We
have two cases:
\begin{enumerate}
\item $\pi$ contains states different from $(\emptyset,R)$, for
  any $R\subseteq Q$, it reaches state $\bot$ and then  loops in this state
  forever, or
\item $\pi$ contains a cycle passing through states which are all different from
  $(\emptyset,R)$ or $\bot$, for any $R\subseteq Q$.
\end{enumerate}

For the first case, let $\alpha$ be the length of the maximal prefix
of $\pi$ containing only states different from $\bot$. 
Let $t(x_k)=((R_1^k,R_2^k),c_A^k)$, for any $0\leq k< \alpha$. 
By the definition of $t$, we have that
$\sigma(\lambda_A(R_2^0)\ldots \lambda_A(R_2^k))=c_A^k$, for any
$0\leq k<\alpha$.

Let $\rho'=\big((q_{k-1},R_{k-1})\sdup{c_k}(q_k,R_k)\big)_{k\geq 1}$ be an infinite run in
$\widehat {\Gamma_A}$ such that:
\begin{itemize}
\item $\rho'[0..i] \sim_A \rho[0..i]$ and $q_i\in R_1^0$,
\item $q_{i+k}\in R_1^k$ and $R_{i+k}=R_2^k$, for all $\alpha> k\geq 1$. 
\item note that, by definition of $\pi$, $\breve
  \delta((R_1^{\alpha-1},R_2^{\alpha-1}),c_A^{\alpha-1})=\bot$.
  We define $(q_{i+\alpha},R_{i+\alpha})\in \widehat Q$ such that
  $(q_{i+\alpha-1},R_2^{\alpha-1})\sdup{c}(q_{i+\alpha},R_{i+\alpha})
  \in \widehat\delta$, for some $c\in C$ with
  $c\restr{A}=c_A^{\alpha-1}$, and $\{p_1,p_2\} \cap
  \lambda(q_{i+\alpha})=\emptyset$.
\end{itemize}
By definition of $t$, this run exists and it is compatible with $\sigma$.
Also, starting with position $i$, $\rho'$ contains a sequence of
states labeled by $p_1$ but not by $p_2$ followed by a state which is
not labeled by $p_1$ or $p_2$.  Consequently, $(\widehat
{\Gamma_A},\rho',i) \not \models p_1 \untl p_2$ which contradicts the
hypothesis.

Similarly, for the second case above, we can construct a run $\rho'$
in $\widehat {\Gamma_A}$ compatible with $\sigma$ such that
$\rho'[0..i] \sim_A \rho[0..i]$ and $(\widehat {\Gamma_A},\rho',i)
\not\models \dmnd p_2$. Consequently, $(\widehat {\Gamma_A},\rho',i)
\not\models p_1 \untl p_2$ which contradicts the hypothesis.

($\Leftarrow$)
Assume that $t$ is a tree accepted by $\breve{\AAA_{\widehat q}}$.
We will construct inductively a strategy $\sigma$ which 
is compatible with $\rho[0..i]$ and 
satisfies the required conditions for witnessing 
that $(\widehat{\Gamma_A},\widehat\rho,i) \models \ceoal p_1 \untl p_2$.

Suppose that the run $\rho$ is $\rho = (\widehat {q_{j-1}} \sdup{c_j} \widehat{q_j})_{j\geq 1}$.
First, we may define $\sigma$ for sequences of elements in $2^{Prop_A}$ of length at most $i$:
for any $A$-history of length less than or equal to $i$, 
$w \in \big(2^{Prop_A}\big)^*$, $|w| = j$ with $1\leq j\leq i$, we put
\[
\sigma(w) = 
\begin{cases}
c_j\restr{A} & \text{ if  } w = \widehat{\lambda_A}(\widehat{q_0})\ldots \widehat{\lambda_A}(\widehat{q_{j-1}}) \\
\text{arbitrary} & \text{ otherwise}
\end{cases}
\]
For defining $\sigma$ on sequences of length greater than $i$, let's 
denote first $w_{\widehat \rho} = \widehat{\lambda_A}(\widehat{q_0})\ldots \widehat{\lambda_A}(\widehat{q_{i-1}})$.
Also, given a sequence of subsets of $Prop_A$, 
$z = (Z_1\cdot \ldots \cdot Z_k) \in (2^{Prop_A})^*$ and a node $x \in \supp(t)$, 
we say that $z$ \emph{labels a path from $\eps$ to $x$ in $t$} if 
the $A$-history along the path from $\eps$ to $x$ in $t$ is exactly $z$, that is,
\[
\forall y \preceq x, \forall 0\leq j\leq |x|, \text{ if } |y| = j \text{ then } \widehat{\lambda_A}\big(t(y)\restr{1}\big) = Z_{j+1}.
\]

Then, for all $k\geq 1$ we put:
\[
\sigma(w_{\widehat \rho}\, Z_1\ldots Z_k) = 
\begin{cases}
t(x)\restr{2} & \text{if }  (Z_1,\ldots,Z_k) \text{ labels a path from $\eps$ to $x$ in } t \\
\text{arbitrary} & \text{otherwise}
\end{cases}
\]

To prove that $\sigma$ is a strategy that witnesses for $(\widehat {\Gamma_{A}}, \widehat \rho, i) \models \ceoal p_1 \untl p_2$,
take some run $\rho'$ compatible with $\sigma$ and for which 
$\rho'[0..i]\sim_A \rho[0..i]$.
We may prove that, if we denote the run as $\rho'= (\widehat {q'_{j-1}} \sdup {c'_j} \widehat {q'_j})_{j\geq 1}$,
with $\widehat{q'_j} = (r_j,S_j) \in \widehat Q$, and we also denote 
$Z_j = \widehat {\lambda_A}(\widehat{q'_j})$, then:
\begin{itemize}
\item there exists a path $(x_{j-i})_{j\geq i}$ in $t$ with $t(x_{j-i})\restr{1} = (R_1^j,S_{j})$ and $t(x_0)\restr{1} = (R_1^0,S_i)$, for some $R_1^k\subseteq Q$, $0\leq k$.
\item for all $j \geq i+1$, $c_j'\restr{A} = \sigma(Z_0\ldots Z_{j-1}) = t(x_{j-i-1})\restr{2}$.
\end{itemize}
This property follows by induction on $j$, and ends the proof of our theorem.
\end{proof}

\paragraph{Case $\ceoal p_1 \wkuntl p_2$:} 
The construction is almost entirely the same as for the previous case, the only difference being the accepting condition.
For this case, the condition from the until case is relaxed: 
any path of an accepting tree may still only have labels of the type $(R_1,R_2)$
denoting the fact that all the runs that are simulated by the path and lead to 
a member of $R_1$ are only labeled with $p_1$.
But we no longer require that, on each path, 
a label of the type $(\emptyset,R)$ occurs. 
This is due to the fact that $p_1 \wkuntl p_2$ 
does not incorporate the obligation to reach a point where $p_2$ holds,
runs on which $p_1$ holds forever are also acceptable.

So, formally, the construction for $\ceoal p_1 \wkuntl p_2$ is the following:
$\widetilde \AAA_{(q,S)} = (\breve Q, C_A, \breve \delta, \breve {Q}_0, \widetilde \FFF)$ 
where $\breve Q, \breve Q_0$ and $\breve \delta$ are the same as in the construction for 
$\ceoal p_1 \wkuntl p_2 $, while the acceptance condition is the following:
\[
\widetilde \FFF = \big\{ \RRR \mid \RRR \subseteq \breve Q \big\}.
\]


The following result can be proved similarly to Proposition \ref{PU}.


\begin{proposition}\label{PW}
For any run $\widehat \rho \in \runsinf(\widehat{\Gamma_A})$ and position $i$ on the run for which $\rho[i] = \widehat q = (q,S)$,
\[
(\widehat{\Gamma_A},\rho,i) \models \ceoal p_1 \wkuntl p_2 \text{ if and only if } 
L(\widetilde\AAA_{\widehat q}) \neq \emptyset
\]
\end{proposition}

\setlength{\unitlength}{1.6pt}

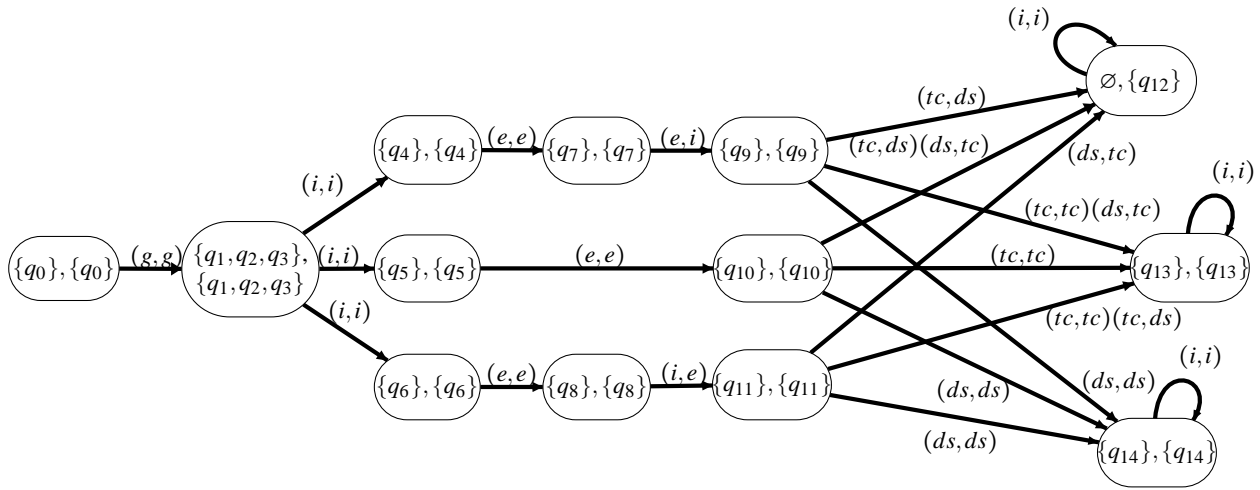
\begin{figure}
\hspace{-8mm}
\hspace{-1cm}
\begin{picture}(327,134)(0,-154)
\node[NLangle=0.0,Nw=26.11,Nh=15.17,Nmr=7.58](n1)(36.87,-95.6){{\footnotesize $\{q_0\},\{q_0\}$}}

\node[NLdist=3.5,Nw=32.46,Nh=22.58,Nmr=11.29](n12)(80.96,-95.78){{\footnotesize $\{q_1,q_2,q_3\}$,}}
\nodelabel[NLdist=-3.5](n12){{\footnotesize $\{q_1,q_2,q_3\}$}}

\drawedge[linewidth=1.0,ELdist=0.0,AHangle=32.62,AHLength=2.97](n1,n12){{\footnotesize $(g,g)$}}


\node[NLangle=0.0,Nw=27.52,Nh=15.87,Nmr=7.93](n14)(203.73,-67.73){{\footnotesize $\{q_9\},\{q_9\}$}}

\node[NLangle=0.0,Nw=28.11,Nh=15.87,Nmr=7.93](n15)(204.44,-95.6){{\footnotesize $\{q_{10}\},\{q_{10}\}$}}

\node[NLangle=0.0,Nw=28.11,Nh=16.58,Nmr=8.29](n16)(204.26,-123.12){{\footnotesize $\{q_{11}\},\{q_{11}\}$}}

\node[NLangle=0.0,Nw=26.11,Nh=16.58,Nmr=8.29](n17)(291.57,-51.15){{\footnotesize $\emptyset,\{q_{12}\}$}}

\node[NLangle=0.0,Nw=28.11,Nh=16.23,Nmr=8.03](n18)(303.04,-95.43){{\footnotesize $\{q_{13}\},\{q_{13}\}$}}

\node[NLangle=0.0,Nw=28.11,Nh=16.23,Nmr=8.11](n19)(295.27,-139.17){{\footnotesize $\{q_{14}\},\{q_{14}\}$}}

\drawedge[linewidth=1.0,ELdist=0.0,AHangle=27.47,AHLength=2.82](n14,n17){{\footnotesize $(tc,ds)$}}

\drawedge[linewidth=1.0,ELpos=75,ELdist=0.0,AHangle=26.57,AHLength=2.68](n14,n18){{\footnotesize $(tc,tc)(ds,tc)$}}

\drawedge[linewidth=1.0,ELpos=85,ELdist=0.0,AHangle=27.03,AHLength=2.75](n14,n19){{\footnotesize $(ds,ds)$}}

\drawedge[linewidth=1.0,ELpos=60,ELdist=0.0,AHangle=26.57,AHLength=2.57](n15,n18){{\footnotesize $(tc,tc)$}}

\drawedge[linewidth=1.0,ELpos=55,ELdist=-12.0,AHangle=23.86,AHLength=2.84](n15,n19){{\footnotesize $(ds,ds)$}}

\drawedge[linewidth=1.0,ELpos=84,ELdist=-15.0,AHangle=27.05,AHLength=2.64](n16,n17){{\footnotesize $(ds,tc)$}}

\drawedge[linewidth=1.0,ELpos=80,ELdist=-14.0,AHangle=25.2,AHLength=2.82](n16,n18){{\footnotesize $(tc,tc) (tc,ds)$}}

\drawedge[linewidth=1.0,ELdist=-10.0,AHangle=22.78,AHLength=2.71](n16,n19){{\footnotesize $(ds,ds)$}}

\drawedge[linewidth=1.0,ELpos=45,ELdist=0.0,AHangle=30.47,AHLength=2.96](n15,n17){{\footnotesize $(tc,ds)(ds,tc)$}}

\node[NLangle=0.0,Nw=25.4,Nh=16.23,Nmr=8.03](n29)(122.94,-67.56){{\footnotesize $\{q_4\},\{q_4\}$}}

\node[NLangle=0.0,Nw=25.05,Nh=15.52,Nmr=7.76](n30)(122.77,-123.65){{\footnotesize $\{q_6\},\{q_6\}$}}

\node[NLangle=0.0,Nw=25.05,Nh=15.87,Nmr=7.93](n33)(122.77,-95.6){{\footnotesize $\{q_5\},\{q_5\}$}}

\drawedge[linewidth=1.0,ELdist=0.0,AHangle=27.05,AHLength=2.64](n12,n33){{\footnotesize $(i,i)$}}


\node[NLangle=0.0,Nw=25.4,Nh=15.87,Nmr=7.93](n35)(162.81,-67.73){{\footnotesize $\{q_7\},\{q_7\}$}}

\node[NLangle=0.0,Nw=25.4,Nh=15.52,Nmr=7.76](n36)(162.81,-123.65){{\footnotesize $\{q_8\},\{q_8\}$}}

\drawedge[linewidth=1.0,ELdist=0.0,AHangle=28.95,AHLength=2.69](n29,n35){{\footnotesize $(e,e)$}}

\drawedge[linewidth=1.0,ELdist=0.0,AHangle=26.57,AHLength=2.8](n33,n15){{\footnotesize $(e,e)$}}

\drawedge[linewidth=1.0,ELdist=0.0,AHangle=26.57,AHLength=2.68](n30,n36){{\footnotesize $(e,e)$}}

\drawedge[linewidth=1.0,ELdist=0.0,AHangle=25.64,AHLength=2.77](n35,n14){{\footnotesize $(e,i)$}}

\drawedge[linewidth=1.0,ELdist=0.0,AHangle=25.6,AHLength=2.66](n36,n16){{\footnotesize $(i,e)$}}

\drawedge[linewidth=1.0,AHangle=27.01,AHLength=2.86](n12,n29){{\footnotesize $(i,i)$}}

\drawedge[linewidth=1.0,ELdist=-0.5,AHangle=25.64,AHLength=2.77](n12,n30){{\footnotesize $(i,i)$}}



\drawloop[linewidth=1.0,AHangle=23.5,AHLength=2.51,loopdiam=10.49,loopangle=151.14](n17){{\footnotesize $(i,i)$}}

\drawloop[linewidth=1.0,AHangle=23.5,AHLength=2.51,loopdiam=9.48,loopangle=65.46](n18){{\footnotesize $(i,i)$}}

\drawloop[linewidth=1.0,AHangle=23.5,AHLength=2.51,loopdiam=9.48,loopangle=65.46](n19){{\footnotesize $(i,i)$}}
\end{picture}
\caption{A tree automaton for the game arena in Figure \ref{fig:ex}}
\label{fig:ta}
\end{figure}

\begin{example}
For our running example,
the tree automaton constructed from the arena $\widehat{\Gamma_{Alice,Bob}}$ 
(given in Example \ref{ex:ga}), 
for the state $(q_0,\{q_0\})$ and the formula $\phi_1= \cefff[\{Alice,Bob\}] (c\wedge s)$ 
is pictured in Figure \ref{fig:ta}. 
Note that it accepts an infinite tree such that all its paths contain the state $(\emptyset,\{q_{12}\})$ 
but never reach $\bot$. Moreover, this tree defines a strategy for the coalition $\{Alice,Bob\}$ to reach the goal $c\wedge s$.
\end{example}

\subsection{The model-checking algorithm}

Our algorithm for the
model-checking problem for \ATLK{} works by structural induction on
the formula $\phi$ to be model-checked. The input of the algorithm is
a game arena $\Gamma=(Q,C,\delta,Q_0,Prop,\lambda)$ and an enumeration
$\Phi = \{\phi_1,\ldots,\phi_n\}$ of the subformulas of $\phi$, in
which $\phi=\phi_n$ and $\phi_i$ is a subformula of $\phi_j$ only if
$i<j$.  The algorithm determines whether $\phi$ holds at all the
initial states of $\Gamma$.  It works by constructing a sequence of
arenas $\Gamma_k = (Q_k,C,\delta_k,Q_0^k,Prop_k,\lambda_k)$,
$k=0,\ldots,n,$ with $\Gamma_0=\Gamma$. The formula $\phi$ is assumed
to be written in terms of the agents from $Ag$ and the atomic
propositions from $Prop=\bigcup_{a\in Ag}Prop_a$ of $\Gamma$. The atomic
propositions of $\Gamma_1,\ldots,\Gamma_n$ include those of $\Gamma$
and $n$ fresh atomic propositions $p_{\phi_k}$, $k=1,\ldots,n$, which
represent the labelling of the states of these arenas by the
corresponding formulas $\phi_k$. For any $1\leq k\leq n$, upon step
$k$ the algorithm constructs $\Gamma_k$ from $\Gamma_{k-1}$ and
calculates the labelling of its states with formula $\phi_k$.
$Prop_k=Prop\cup\Phi_k$ where $\Phi_k$ denotes $\{p_{\phi_1},\ldots,p_{\phi_k}\}$, $k=0,\ldots,n$. 
The state labelling function $\lambda_k$ is defined so that equivalence 
$p_{\phi_k}\Leftrightarrow\phi_k$ is valid in $\Gamma_{k}$. 
Therefore, we define the formula
$\chi_k=\phi_k[\phi_{k-1}/p_{\phi_{k-1}}],\ldots,[\phi_1/p_{\phi_1}]$
which has at most one connective of the form $\cenxt$,
$\ceoal\untl$, $\ceoal \wkuntl$ or $K_A$. The algorithm computes the states that should
be labeled by $p_{\phi_k}$ using the formula $\chi_k$ which is
equivalent to $\phi_k$.
The fresh propositions $p_{\phi_1},\ldots,p_{\phi_n}$ are not assumed
to be observable by any particular agent. Therefore the requirements
$Prop_k=\bigcup\limits_{a\in Ag}Prop_{a,k}$ on arenas are not met by
$\Gamma_1,\ldots,\Gamma_n$, but this is of no consequence.

Let us note the need to switch, at each step, from 
analyzing $\Gamma_k$ to analyzing $\Gamma_{k+1}$.
This is needed as $\Gamma_{k+1}$ only has the 
necessary information about the identically-observable histories,
needed in the semantics of coalition operators.

In case $\phi_k$ is atomic,
$\Gamma_k=(Q_{k-1},C,\delta_{k-1},Q_0^{k-1},Prop_k ,\lambda_k)$ where
$\lambda_k(q)\cap Prop_{k-1}=\lambda_{k-1}(q)$ and
$p_{\phi_k}\in\lambda_k(q)$ iff $\phi_k\in\lambda_{k-1}(q)$. In case
$\phi_k$ is not atomic, the construction of $\Gamma_k$ depends on the
main connective of $\chi_k$:
\begin{enumerate}
\item Let $\chi_k$ be a boolean combination of atoms from
  $Prop_{k-1}$. Then $\Gamma_k =
  (Q_{k-1},C,\delta_{k-1},Q_0^{k-1},Prop_k ,\lambda_k)$ where
  $\lambda_k(q)\cap Prop_{k-1}=\lambda_{k-1}(q)$ and
  $p_{\phi_k}\in\lambda_k(q)$ iff the boolean formula $\bigwedge_{p\in
    \lambda_{k-1}(q)} p$ implies $\chi_k$.
\item Let $\chi_k$ be $K_A p$ for some $p\in Prop_{k-1}$. Consider the arena 
$\widehat {(\Gamma_{k-1})}_A$ defined as in Subsection \ref{sec:031}. Then
  $\Gamma_k = (\widehat Q_{k-1},C,\widehat\delta_{k-1},\widehat
  Q_0^{k-1}, Prop_k ,\lambda_k)$ where
  $\lambda_k(q)\cap Prop_{k-1}=\widehat\lambda_{k-1}(q)$ and
  $p_{\phi_k}\in \lambda_k(q)$ iff $q\in \widehat Q_{k-1}^{K_A p}$,
  where $\widehat Q_{k-1}^{K_A p}$ is defined in (\ref{id:q-hat}).
\item Let $\chi_k$ be $\cenxt p$ for some $p\in Prop_{k-1}$. Consider
  $\widehat{(\Gamma_{k-1})_A}$. Then $\Gamma_k = (\widehat
  Q_{k-1},C,\widehat \delta_{k-1},\widehat Q_0^{k-1},Prop_k
  ,\lambda_k)$ where
  $\lambda_k(q)\cap Prop_{k-1}=\widehat\lambda_{k-1}(q)$ and
  $p_{\phi_k}\in \lambda_k(q)$ iff $q \in \widehat Q_{k-1}^{\cenxt
    p}$, where $\widehat Q_{k-1}^{\cenxt p}$ is defined in
  (\ref{id:next}).
\item Let $\chi_k$ be $\ceoal p_1 \untl p_2$ for some
  $p_1,p_2\in Prop_{k-1}$. Consider $\widehat{(\Gamma_{k-1})_A}$ again
  and, for each state $\widehat q \in \widehat Q_{k-1}$, construct the
  tree automaton $\breve\AAA_{\widehat q}$.  Then put $\Gamma_k =
  (\widehat Q_{k-1},C,\widehat \delta_{k-1},\widehat Q_0^{k-1},
  Prop_k,\lambda_k)$ where
  $\lambda_k(q)\cap Prop_{k-1}=\widehat\lambda_{k-1}(q)$ and
  $p_{\phi_k}\in \lambda_k(q)$ iff $L(\breve\AAA_{\widehat q}) \neq
  \emptyset$.
\item Finally, let $\chi_k$ be $\ceoal p_1 \wkuntl p_2$ for some
  $p_1,p_2\in Prop_{k-1}$. Consider $\widehat{(\Gamma_{k-1})_A}$ again
  and, for each state $\widehat q \in \widehat Q_{k-1}$, construct the
  tree automaton $\widetilde\AAA_{\widehat q}$.  Then put $\Gamma_k =
  (\widehat Q_{k-1},C,\widehat \delta_{k-1},\widehat Q_0^{k-1},
  Prop_k,\lambda_k)$ where
  $\lambda_k(q)\cap Prop_{k-1}=\widehat\lambda_{k-1}(q)$ and
  $p_{\phi_k}\in \lambda_k(q)$ iff $L(\widetilde\AAA_{\widehat q}) \neq
  \emptyset$.
\end{enumerate}

The following result is a direct consequence of Propositions \ref{P1}, \ref{PU}, and \ref{PW}.

\begin{theorem}
Let $\Gamma_n =(Q_n,C,\delta_n,Q_0^n,Prop_n ,\lambda_n)$ be the last game arena obtained in the algorithm described above. Then,
\[
p_{\phi}\in \lambda_n(q),\mbox{ for all states }q\in Q_0^n\quad \mbox{ iff }\quad(\Gamma,\rho,0)\models\phi,\mbox{ for all runs }\rho \in \runsinf(\Gamma).
\]
\end{theorem}

\section{Concluding remarks}

We have presented a model-checking technique for \ATLK{}, a variant of the Alternating Temporal Logic with Knowledge, 
in which coalitions may coordinate their actions, based on their distributed knowledge of the system state. 
The technique is based on a state labeling algorithm which involves tree automata for identifying states 
to be labeled with cooperation modality subformulas, and a state splitting construction which 
serves for identifying (finite classes of) histories which are indistinguishable to some coalition.

According to our semantics, while distributed knowledge is used for constructing coalition strategies, 
it is assumed that the individual agents in the coalition gain no access to that knowledge 
as a side effect of their cooperation. 
That is why the proposed semantics corresponds to coalitions being organised under {\em virtual supervisors} 
who guide the implementation of strategies by receiving reports on observations of the coalitions' 
members and, in return, just directing the members' actions without making any other knowledge available to them.

The possibility of a subsequent increase of individual knowledge as a side effect 
of the use of distributed knowledge for coordinated action, 
which we avoid by introducing virtual supervisors, 
becomes relevant only in settings such as that of ATL with incomplete information. 
This possibility appears to be an interaction between the understanding of distributed knowledge 
as established in non-temporal epistemic logic and temporal settings. 
This is just one of the numerous subtle interpretation issues which were created by the straightforward 
introduction of modalities from non-temporal epistemic logic and cooperation modalities into temporal logics. 
For an example of another such issue,  a semantics for ATL in which agents, 
once having chosen a strategy for achieving a certain main goal, cannot revise it upon 
considering the reachability of subgoals, was proposed and studied in \cite{goranko-irrevocable2007}. 

The state labeling algorithm suggests that tree automata with partial observations 
and with partially-observable objectives might be useful to study.
We believe that the two state-labeling constructions can be generalized to such automata, 
giving us also a decision method for the ``starred'' version of \ATLK{}.

\subsection*{Acknowledgments}

Ferucio Tiplea, Ioana Boureanu and Sorin Iftene have participated in some 
preliminary discussions on the subject of this paper and have suggested several corrections 
that are included here.

\bibliographystyle{alpha}
\bibliography{epistemic,dima}

\begin{thebibliography}{vdHLW06}

\bibitem[{\AA}GJ07]{goranko-irrevocable2007}
Th. {\AA}gotnes, V.~Goranko, and W.~Jamroga.
\newblock Alternating-time temporal logics with irrevocable strategies.
\newblock In {\em Proceedings of TARK'07}, pages 15--24, 2007.

\bibitem[AHK98]{alur-atl1997}
R.~Alur, Th. Henzinger, and O.~Kupferman.
\newblock Alternating-time temporal logic.
\newblock In {\em Proceedings of COMPOS'97}, volume 1536 of {\em LNCS}, pages
  23--60. Springer Verlag, 1998.

\bibitem[AHK02]{alur-atl-jacm2002}
R.~Alur, Th. Henzinger, and O.~Kupferman.
\newblock Alternating-time temporal logic.
\newblock {\em Journal of the ACM}, 49(5):672--713, 2002.

\bibitem[BJ09]{bulling-harder-2009}
N.~Bulling and W.~Jamroga.
\newblock Model checking {ATL$^+$} is harder than it seemed.
\newblock Technical Report IfI-09-13, 2009.

\bibitem[CDHR06]{chatterjee-imperfect-information2006}
K.~Chatterjee, L.~Doyen, Th.A. Henzinger, and J.-F. Raskin.
\newblock Algorithms for omega-regular games with imperfect information.
\newblock In {\em Proceedings of CSL'06}, volume 4207 of {\em Lecture Notes in
  Computer Science}, pages 287--302. Springer, 2006.

\bibitem[Dim08]{dima08clima}
C.~Dima.
\newblock Revisiting satisfiability and model-checking for {CTLK} with
  synchrony and perfect recall.
\newblock In {\em Proceedings of the 9th International Workshop on
  Computational Logic in Multi-Agent Systems (CLIMA IX)}, volume 5405 of {\em
  LNAI}, pages 117--131, 2008.

\bibitem[Dim10]{dima-jlc2010}
C.~Dima.
\newblock Non-axiomatizability for linear temporal logic of knowledge with
  concrete observability, 2010.
\newblock Submitted.

\bibitem[FHMV04]{FaginHalpernVardi}
R.~Fagin, J.~Halpern, Y.~Moses, and M.~Vardi.
\newblock {\em Reasoning about knowledge}.
\newblock The MIT Press, 2004.

\bibitem[GD08]{guelev-dima2008dalt}
D.P. Guelev and Catalin Dima.
\newblock Model-checking strategic ability and knowledge of the past of
  communicating coalitions.
\newblock In {\em Proceedings of DALT 2008}, volume 5397 of {\em Lecture Notes
  in Computer Science}, pages 75--90. Springer, 2008.

\bibitem[GJ04]{goranko-jamroga-comparing-2004}
V.~Goranko and W.~Jamroga.
\newblock Comparing semantics of logics for multi-agent systems.
\newblock {\em Synthese}, 139(2):241--280, 2004.

\bibitem[GvD06]{goranko-drimmelen06}
V.~Goranko and G.~van Drimmelen.
\newblock Complete axiomatization and decidability of alternating-time temporal
  logic.
\newblock {\em Theoretical Computer Science}, 353(1-3):93--117, 2006.

\bibitem[JA07]{jamroga-de-dicto2007}
W.~Jamroga and Th. Agotnes.
\newblock Constructive knowledge: What agents can achieve under imperfect
  information.
\newblock {\em Journal of Applied Non-Classical Logics}, 17(4):423--475, 2007.

\bibitem[KP05]{kacprzak-penczek-aamas-2005}
M.~Kacprzak and W.~Penczek.
\newblock Fully symbolic unbounded model checking for alternating-time temporal
  logic.
\newblock {\em Autonomous Agents and Multi-Agent Systems}, 11(1):69--89, 2005.

\bibitem[LMO08]{oreiby-2008}
Fr. Laroussinie, N.~Markey, and Gh. Oreiby.
\newblock On the expressiveness and complexity of {ATL}.
\newblock {\em Logical Methods in Computer Science}, 4(2), 2008.

\bibitem[Sch04]{schobbens-atl-ir2004}
P.-Y. Schobbens.
\newblock Alternating-time logic with imperfect recall.
\newblock {\em Electronic Notes in Theoretical Computer Science}, 85(2):82--93,
  2004.

\bibitem[Tho97]{thomas97handbook}
W.~Thomas.
\newblock Languages, automata, and logic.
\newblock In G.~Rozenberg and A.~Salomaa, editors, {\em Handbook of Formal
  Languages}, volume 3, Beyond Words, pages 389--455. Springer Verlag, 1997.

\bibitem[vdHLW06]{lomuscio-practicalATL2006}
W.~van~der Hoek, A.~Lomuscio, and M.~Wooldridge.
\newblock On the complexity of practical {ATL} model checking.
\newblock In {\em Proceedings of AAMAS 2006}, pages 201--208. ACM, 2006.

\bibitem[vdHW03]{wiebe-ATEL2003}
W.~van~der Hoek and M.~Wooldridge.
\newblock Cooperation, knowledge, and time: Alternating-time temporal epistemic
  logic and its applications.
\newblock {\em Studia Logica}, 75(1):125--157, 2003.

\end{thebibliography}
\nocite{goranko-drimmelen06}

\end{document}